\documentclass[aps,pra,twocolumn,nofootinbib,superscriptaddress]{revtex4-2}

\usepackage{amsmath,amssymb,amsthm,mathtools,bm,mathrsfs}
\usepackage{graphicx}
\usepackage{enumitem}
\usepackage[colorlinks=true,citecolor=blue,linkcolor=blue,urlcolor=blue]{hyperref}

\allowdisplaybreaks

\theoremstyle{plain}
\newtheorem{theorem}{Theorem}
\newtheorem{proposition}{Proposition}
\newtheorem{corollary}{Corollary}

\theoremstyle{definition}
\newtheorem{definition}{Definition}
\newtheorem{remark}{Remark}

\newcommand{\Tr}{\operatorname{Tr}}
\newcommand{\diag}{\operatorname{diag}}

\newcommand{\C}{\mathbb C}
\newcommand{\D}{\mathcal D}
\newcommand{\W}{\mathcal W}
\newcommand{\Id}{\mathbb I}
\newcommand{\Gap}{\Gamma}
\newcommand{\BG}{\mathcal G_{\mathrm B}}

\begin{document}

\title{A low order Bargmann invariant hierarchy for set coherence}

\author{Yan-Ling Wang}
\email{wangylmath@yahoo.com}
\affiliation{School of Computer Science and Technology, Dongguan University of Technology, Dongguan 523808, China}

\begin{abstract}
Set coherence is a basis-independent relational form of quantum coherence: a
finite family of quantum states is set incoherent exactly when all its members
are diagonal in one common basis. We determine how much low-order Bargmann data
are needed to decide this property. For two states, second-order data are
complete for qubits but fail for qutrits, while complete third-order data are
sufficient for qutrits but fail already in dimension four. We then show that
fourth-order, ordering-sensitive Bargmann invariants give the first universal
pairwise criterion for set coherence. Applied to all unordered pairs, this
criterion yields a complete test for arbitrary finite families. The result
provides a low-order hierarchy connecting cyclic trace invariants with the
noncommutativity that prevents a common incoherent basis.
\end{abstract}

\maketitle

\section{Introduction}

Quantum coherence is a basic signature of quantum behavior and a central
resource in quantum information science. In the standard resource theory,
coherence is defined relative to a fixed reference basis~\cite{Baumgratz2014,Streltsov2017}. This basis may be selected by a Hamiltonian,
a measurement device, or a computational encoding. Without such a prescribed
basis, however, a single density operator has no intrinsic basis-independent
coherence: it can always be diagonalized in its own eigenbasis.

A nontrivial basis-independent notion appears when several states are considered
at the same time. Designolle \emph{et al.} introduced this relational resource
under the name of set coherence~\cite{Designolle2021}. Given a finite family
\[
\vec\rho=(\rho_1,\ldots,\rho_n)
\]
of states on \(\C^d\), each individual state can be diagonalized in some basis.
The question is whether all of them can be diagonalized in the same basis. If
such a common basis exists, the family is set incoherent; otherwise it is set
coherent. Thus set coherence does not belong to any single member of the
family. It records the incompatibility among the eigenbasis structures of the
states.

With full knowledge of the density matrices, deciding set coherence is
straightforward. A finite family of Hermitian matrices is simultaneously
diagonalizable if and only if it is pairwise commuting. The question studied
here is more restrictive and more invariant in spirit: can this relational
nonclassicality be decided from a prescribed list of basis-independent 
data, without reconstructing the states and without choosing a reference basis?

The basis-independent data considered in this work are cyclic traces of products of the
states. We refer to them as generalized Bargmann invariants. For a word
\(w=i_1i_2\cdots i_m\) in the state labels, the corresponding invariant is
\[
\Delta_w(\vec\rho)=\Tr(\rho_{i_1}\rho_{i_2}\cdots\rho_{i_m}).
\]
These quantities are invariant under simultaneous unitary conjugation and are
therefore natural for basis-independent certification. Bargmann-type invariants
and related multitime trace quantities appear in geometric phases,
Kirkwood--Dirac quasiprobabilities, weak values, coherence and overlap
inequalities, multiphoton indistinguishability, quantum imaginarity, quantum entanglement, and
relational information between states~\cite{Bargmann1964,Kirkwood1933,Dirac1945,ArvidssonShukur2024,CanImaginarityBroadcast2024,Schmid2024KD,Dressel2015,WagnerGalvao2023WeakValues,GalvaoBrod2020,WagnerBarbosaGalvao2024,Giordani2021,Giordani2023,Fernandes2024,LiTan2025, ZhangXieLi2025,Zhang25b,Oszmaniec2024,LiWagnerZhang2026}. They are also
operationally meaningful, since multivariate traces can be estimated by
multi-copy, interferometric, or randomized-measurement-inspired
protocols~\cite{Wagner2024QuantumCircuits,Quek2024,Simonov2025,Shin2025}.

A recent systematic language for such data is provided by Bargmann
scenarios~\cite{Wagner2026BargmannScenarios}. A scenario fixes a finite set of
words and studies the attainable tuple of the corresponding Bargmann
invariants. In this language, one can compare quantum-realizable tuples with
those generated by simultaneously diagonal families. Points that cannot arise
from a common diagonal basis witness a relational resource. This viewpoint
leads to the guiding question of the paper: 
\emph{How low in Bargmann order can one go while still deciding set coherence?}

At first sight, a complete answer might be expected to require many trace
polynomials. Indeed, invariant-theoretic descriptions of simultaneous unitary
orbits of matrices are usually generated by large families of trace
polynomials~\cite{Procesi1976,Razmyslov1974,Formanek1986}. Set coherence is a
more targeted property. It does not ask for the full unitary orbit of a family,
but only whether the family is simultaneously diagonalizable. This distinction
leaves open the possibility that very low-order Bargmann data may already be
complete in some dimensions.  The purpose of this paper is to determine how far such low-order data can go
and to identify the first universal pairwise order.

The paper is organized as follows. Section~\ref{sec:setup} introduces set
coherence,  and completeness of a Bargmann
scenario. Section~\ref{sec:hierarchy} develops the two-state hierarchy  from the qubit result to the four-dimensional counterexample.
Section~\ref{sec:fourth} proves the universal fourth-order criterion and then
interprets it through the Bargmann gap. Section~\ref{sec:conclusion} concludes
with possible extensions.  

\section{Preliminaries}
\label{sec:setup}

Let \(\D(\C^d)\) denote the set of density operators on \(\C^d\). We write
\(\Tr\) for the trace, \(\Id_d\) for the identity operator on \(\C^d\), and
\([A,B]=AB-BA\) for the commutator. The Hilbert--Schmidt norm is
\[
\|X\|_2=\sqrt{\Tr(X^\dagger X)}.
\]

\begin{definition}[Set coherence, \cite{Designolle2021}]
A finite family \(\vec\rho=(\rho_1,\ldots,\rho_n)\), with
\(\rho_i\in\D(\C^d)\), is \emph{set incoherent} if there exists a unitary
\(U\) such that all \(U\rho_iU^\dagger\) are diagonal. If no such unitary
exists, the family is \emph{set coherent}.
\end{definition}

Since density operators are Hermitian, the standard simultaneous
diagonalization criterion gives
\[
\vec\rho\ \text{is set incoherent}
\quad\Longleftrightarrow\quad
[\rho_i,\rho_j]=0\quad\text{for all }i<j.
\]

For a word \(w=i_1i_2\cdots i_m\) in the labels \(1,\ldots,n\), define the
generalized Bargmann invariant
\[
\Delta_w(\vec\rho)=\Tr(\rho_{i_1}\rho_{i_2}\cdots\rho_{i_m}).
\]
Cyclically equivalent words give the same invariant.

\begin{definition}[Bargmann scenario, \cite{Wagner2026BargmannScenarios}]
A Bargmann scenario is a finite set \(\W\) of words with its defining value map
\[
\Phi_\W(\vec\rho)=\bigl(\Delta_w(\vec\rho)\bigr)_{w\in\W}.
\]
\end{definition}
For fixed dimension \(d\), let
\begin{align}
B_d(\W)
&=\{\Phi_\W(\vec\rho):\vec\rho\in\D(\C^d)^n\},\\
C_d(\W)
&=\{\Phi_\W(\vec\rho):\vec\rho\in\D(\C^d)^n
\ \text{is set incoherent}\},\\
I_d(\W)
&=\{\Phi_\W(\vec\rho):\vec\rho\in\D(\C^d)^n
\ \text{is set coherent}\}.
\end{align}
We say that \(\W\) \emph{decides set coherence} in dimension \(d\), or is
\emph{complete} in dimension \(d\), if
\[
C_d(\W)\cap I_d(\W)=\varnothing.
\]
Equivalently, no set-incoherent and set-coherent families in the same dimension
have the same \(\W\)-data.

\section{The Low-Order Two-State Hierarchy}
\label{sec:hierarchy}

For two states \((\rho,\sigma)\), we use the scenarios
\begin{align*}
\W_2&=\{11,22,12\},\\
\W_3&=\{11,22,12,111,222,112,122\},\\
\W_4&=\{1122,1212\}.
\end{align*}
The hierarchy proved below is summarized in Table~\ref{tab:hierarchy}. The
entry ``complete'' means that the corresponding data separate commuting from
noncommuting pairs in the given dimension.

\begin{table}[t]
\caption{Completeness of the two-state Bargmann scenarios considered in this
work. Complete means \(C_d(\W)\cap I_d(\W)=\varnothing\).}
\label{tab:hierarchy}
\begin{ruledtabular}
\begin{tabular}{cccc}
Dimension & \(\W_2\) & \(\W_3\) & \(\W_4\)\\
\hline
\(d=2\) & complete & complete & complete\\
\(d=3\) & incomplete & complete & complete\\
\(d\ge4\) & incomplete & incomplete & complete
\end{tabular}
\end{ruledtabular}
\end{table}

We start at second order. For \(\W_2\), write
\[
\Phi_{\W_2}(\rho,\sigma)=(x,y,z),
\]
where
\[
x=\Tr\rho^2,\qquad y=\Tr\sigma^2,\qquad z=\Tr(\rho\sigma).
\]
These three numbers are the Hilbert--Schmidt lengths of the two states and
their Hilbert--Schmidt overlap. For qubits this information is already enough.

\begin{theorem}[Second-order completeness for qubits]
\label{thm:qubitW2}
For two qubit states, the two purities and the Hilbert--Schmidt overlap decide
whether the pair admits a common incoherent basis. More explicitly,
\begin{subequations}
\label{eq:qubit-regions}
\begin{align}
B_2(\W_2)
&=
\Bigl\{(x,y,z):x,y\in[1/2,1],
\nonumber\\[-1mm]
&\hspace{11mm}
(2z-1)^2\le(2x-1)(2y-1)\Bigr\},\\
C_2(\W_2)
&=
\Bigl\{(x,y,z)\in B_2(\W_2):
\nonumber\\[-1mm]
&\hspace{11mm}
(2z-1)^2=(2x-1)(2y-1)\Bigr\},\\
I_2(\W_2)
&=
\Bigl\{(x,y,z)\in B_2(\W_2):
\nonumber\\[-1mm]
&\hspace{11mm}
(2z-1)^2<(2x-1)(2y-1)\Bigr\}.
\end{align}
\end{subequations}
Hence \(C_2(\W_2)\cap I_2(\W_2)=\varnothing\).
\end{theorem}

\begin{proof}
Write the two states in Bloch form,
\[
\rho=\frac{\Id_2+\bm r\cdot\bm\sigma}{2},
\qquad
\sigma=\frac{\Id_2+\bm s\cdot\bm\sigma}{2},
\]
with \(|\bm r|,|\bm s|\le1\). Then
\[
2x-1=|\bm r|^2,\qquad
2y-1=|\bm s|^2,\qquad
2z-1=\bm r\cdot\bm s.
\]
The inequality is Cauchy--Schwarz. Equality holds precisely when the two Bloch
vectors are linearly dependent, including the case where one vector is zero.
This is exactly the condition \([\rho,\sigma]=0\).
\end{proof}

The qubit result has a simple geometric reason: in the Bloch ball,
commutativity of two qubit states is equivalent to collinearity of their Bloch
vectors, and collinearity is detected by the two lengths and the inner product.
The same reasoning does not survive in dimension three. There, second-order
data still determine only the Hilbert--Schmidt lengths and overlap of the two
states, but not the relative eigenbasis information needed to decide
commutativity.

\begin{proposition}[Second-order incompleteness for qutrits]
\label{prop:qutritW2fail}
The two purities and the Hilbert--Schmidt overlap do not decide qutrit set
coherence; equivalently,
\[
C_3(\W_2)\cap I_3(\W_2)\neq\varnothing.
\]
\end{proposition}

\begin{proof}
Consider the commuting pair
\[
\rho_0=\diag\!\left(\frac12,\frac12,0\right),
\qquad
\sigma_0=\diag\!\left(\frac12,0,\frac12\right).
\]
It satisfies
\[
\Phi_{\W_2}(\rho_0,\sigma_0)=\left(\frac12,\frac12,\frac14\right).
\]
Now set
\[
\rho_1=\diag\!\left(\frac12,\frac12,0\right),
\qquad
\sigma_1=
\begin{pmatrix}
1/4&0&1/4\\
0&1/4&0\\
1/4&0&1/2
\end{pmatrix}.
\]
The matrix \(\sigma_1\) is positive semidefinite and has trace one. Moreover,
\[
\Tr\sigma_1^2=\frac12,
\qquad
\Tr(\rho_1\sigma_1)=\frac14.
\]
Thus
\[
\Phi_{\W_2}(\rho_1,\sigma_1)=\Phi_{\W_2}(\rho_0,\sigma_0).
\]
However, \([\rho_1,\sigma_1]\neq0\), since \(\sigma_1\) has a nonzero matrix
element connecting eigenspaces of \(\rho_1\) with different eigenvalues.
\end{proof}

To overcome this second-order limitation, we next consider the enlarged
two-state scenario consisting of all cyclic words of order at most three. For
\(\W_3\), write
\[
\Phi_{\W_3}(\rho,\sigma)=(x,y,z,a,b,c,d),
\]
where
\[
x=\Tr\rho^2,\quad y=\Tr\sigma^2,\quad z=\Tr(\rho\sigma),
\]
and
\[
a=\Tr\rho^3,\quad b=\Tr\sigma^3,
\quad c=\Tr(\rho^2\sigma),\quad d=\Tr(\rho\sigma^2).
\]
The qutrit case is special because the trace, second moment, and third moment
determine the characteristic polynomial of a density operator. The mixed
moments then constrain how the eigenspaces of the two states can be arranged
relative to each other.

\begin{theorem}[Third-order completeness for qutrits]
\label{thm:qutritW3complete}
\[
C_3(\W_3)\cap I_3(\W_3)=\varnothing.
\]
Equivalently, if a qutrit pair has the same \(\W_3\)-tuple as a commuting
qutrit pair, then the pair itself is commuting.
\end{theorem}

\begin{proof}
Suppose that \((\rho,\sigma)\) and a commuting pair \((\rho_0,\sigma_0)\) have
the same \(\W_3\)-tuple. We show that \([\rho,\sigma]=0\).

The equalities
\[
\Tr\rho^2=\Tr\rho_0^2,
\qquad
\Tr\rho^3=\Tr\rho_0^3,
\]
together with trace normalization, imply that \(\rho\) and \(\rho_0\) have the
same characteristic polynomial. Since both are Hermitian, they are unitarily
similar. Applying a simultaneous unitary conjugation, we may assume
\[
\rho=\rho_0=A.
\]
Because \((\rho_0,\sigma_0)\) is commuting, choose a basis
\(\{|e_i\rangle\}_{i=1}^3\) in which
\[
A=\diag(a_1,a_2,a_3),
\qquad
\sigma_0=D=\diag(d_1,d_2,d_3).
\]
It remains to prove that \([A,\sigma]=0\).

First assume that \(A\) has three distinct eigenvalues. Set
\[
s_i=\langle e_i|\sigma|e_i\rangle.
\]
The words \(12\) and \(112\), together with normalization, give
\[
\Tr(A^k\sigma)=\Tr(A^kD),\qquad k=0,1,2.
\]
Equivalently,
\[
\sum_{i=1}^3 a_i^k s_i=\sum_{i=1}^3 a_i^k d_i,
\qquad k=0,1,2.
\]
Since the Vandermonde matrix associated with \(a_1,a_2,a_3\) is invertible,
we obtain \(s_i=d_i\) for all \(i\). The equality
\(\Tr\sigma^2=\Tr D^2\) then gives
\[
\sum_i s_i^2+2\sum_{i<j}|\sigma_{ij}|^2=\sum_i d_i^2.
\]
Hence every off-diagonal entry of \(\sigma\) in the eigenbasis of \(A\)
vanishes, so \([A,\sigma]=0\).

Now assume that \(A\) has eigenvalues \(a,a,c\), with \(a\neq c\). Relabeling
if necessary,
\[
A=a\Id_3+(c-a)P,
\qquad
P=|e_3\rangle\langle e_3|.
\]
From \(\Tr(A\sigma)=\Tr(AD)\) and \(\Tr\sigma=\Tr D\), we obtain
\[
\Tr(P\sigma)=d_3.
\]
From \(\Tr(A\sigma^2)=\Tr(AD^2)\) and \(\Tr\sigma^2=\Tr D^2\), we obtain
\[
\Tr(P\sigma^2)=d_3^2.
\]
Therefore
\[
\langle e_3|\sigma|e_3\rangle=d_3,
\qquad
\langle e_3|\sigma^2|e_3\rangle=d_3^2,
\]
or equivalently
\[
\|(\sigma-d_3\Id_3)|e_3\rangle\|^2=0.
\]
Thus \(|e_3\rangle\) is an eigenvector of \(\sigma\). Since \(\sigma\) is
Hermitian, the orthogonal complement of \(|e_3\rangle\) is invariant under
\(\sigma\). Hence \(\sigma\) is block diagonal with respect to the eigenspaces
of \(A\), and \([A,\sigma]=0\). If \(A\) is scalar, the conclusion is
immediate. This proves the theorem.
\end{proof}

\begin{corollary}
\label{cor:qutrittest}
Let \(v=(x,y,z,a,b,c,d)\in B_3(\W_3)\). The eigenvalues \(p_1,p_2,p_3\) of
\(\rho\) are the roots of
\[
t^3-t^2+\frac{1-x}{2}t-\frac{1-3x+2a}{6}=0,
\]
and the eigenvalues \(q_1,q_2,q_3\) of \(\sigma\) are obtained analogously
from \((y,b)\). Then \(v\in C_3(\W_3)\) if and only if the mixed moments
\(z,c,d\) are compatible with a simultaneous diagonal realization having
spectra \(\{p_i\}\) and \(\{q_i\}\). In the nondegenerate case this amounts to
checking whether there exists \(\pi\in S_3\) such that
\[
z=\sum_i p_iq_{\pi(i)},\qquad
c=\sum_i p_i^2q_{\pi(i)},\qquad
d=\sum_i p_iq_{\pi(i)}^2.
\]
\end{corollary}

\begin{remark}
When degeneracies are present, the permutation test should be understood after
grouping equal eigenvalues. Equivalently, one checks whether there exists a
simultaneously diagonal pair with spectra \(\{p_i\}\) and \(\{q_i\}\) that
realizes the prescribed mixed moments.
\end{remark}

The qutrit theorem is the positive middle step in the hierarchy. Its proof
uses a fact special to three dimensions: the first three moments determine the
spectrum. In dimension four, the same order-\(\le3\) information no longer
controls enough of the relative eigenspace structure. The following example
shows that third order is not universal.

\begin{proposition}[Third-order incompleteness in dimension four]
\label{prop:d4fail}
The full order-\(\le3\) two-state Bargmann data are not complete for set
coherence in dimension four; equivalently,
\[
C_4(\W_3)\cap I_4(\W_3)\neq\varnothing.
\]
\end{proposition}

\begin{proof}
Let
\[
\rho=\frac12\diag(1,1,0,0),
\qquad
\sigma_0=\frac12\diag(1,0,1,0).
\]
Define
\[
|\eta_1\rangle=\frac{|e_1\rangle+|e_3\rangle}{\sqrt2},
\qquad
|\eta_2\rangle=\frac{|e_2\rangle+|e_4\rangle}{\sqrt2},
\]
and
\[
\sigma_1=\frac12\bigl(
|\eta_1\rangle\langle\eta_1|+|\eta_2\rangle\langle\eta_2|
\bigr).
\]
Then \([\rho,\sigma_0]=0\), whereas \([\rho,\sigma_1]\neq0\).

The states \(\rho,\sigma_0,\sigma_1\) are all one half times rank-two
projections. Hence
\[
\Tr\rho^2=\Tr\sigma_0^2=\Tr\sigma_1^2=\frac12,
\qquad
\Tr\rho^3=\Tr\sigma_0^3=\Tr\sigma_1^3=\frac14.
\]
Moreover,
\[
\Tr(\rho\sigma_0)=\Tr(\rho\sigma_1)=\frac14.
\]
Since
\[
\rho^2=\frac12\rho,
\qquad
\sigma_0^2=\frac12\sigma_0,
\qquad
\sigma_1^2=\frac12\sigma_1,
\]
we also have
\[
\Tr(\rho^2\sigma_0)=\Tr(\rho^2\sigma_1)=\frac18,
\qquad
\Tr(\rho\sigma_0^2)=\Tr(\rho\sigma_1^2)=\frac18.
\]
Therefore
\[
\Phi_{\W_3}(\rho,\sigma_0)=\Phi_{\W_3}(\rho,\sigma_1)
=\left(\frac12,\frac12,\frac14,\frac14,\frac14,\frac18,\frac18\right),
\]
although the first pair is commuting and the second is not.
\end{proof}

For every \(d>4\), the same construction can be embedded by taking a direct
sum with a zero block of size \(d-4\). This preserves positivity, trace one,
all trace words, and commutativity or noncommutativity. Hence the failure
persists in every dimension \(d\ge4\). The appendix shows that analogous
lower-order failures persist for many-state families.

\section{Fourth-Order Criterion and the Bargmann Gap}
\label{sec:fourth}

The hierarchy so far has a clear message. Second order is complete only for
qubits; third order reaches qutrits but fails in dimension four. We now show
that the next order succeeds in all finite dimensions.

Consider
\[
\W_4=\{1122,1212\}.
\]
For a pair \((\rho,\sigma)\), write
\[
u=\Tr(\rho^2\sigma^2),
\qquad
v=\Tr(\rho\sigma\rho\sigma).
\]
The two words contain the same two copies of each state. They differ only in
the order in which these copies are multiplied. This ordering sensitivity is
precisely what detects noncommutativity.

\begin{theorem}[Universal fourth-order Bargmann criterion]
\label{thm:fourth-family}
For any finite-dimensional density operators \(\rho\) and \(\sigma\),
\[
\Gap(\rho,\sigma)
:=\Tr(\rho^2\sigma^2)-\Tr(\rho\sigma\rho\sigma)
=\frac12\|[\rho,\sigma]\|_2^2.
\]
Consequently,
\[
\Gap(\rho,\sigma)=0
\quad\Longleftrightarrow\quad
[\rho,\sigma]=0.
\]
Thus \(\W_4\) decides two-state set coherence in every finite dimension. More
generally, a finite family \(\vec\rho=(\rho_1,\ldots,\rho_n)\) is set
incoherent if and only if
\[
\Gap(\rho_i,\rho_j)=0
\qquad\text{for all }i<j.
\]
Equivalently, the pairwise fourth-order scenario
\[
\W_4^{(n)}=\{iijj,ijij:1\le i<j\le n\}
\]
decides set coherence for arbitrary finite families.
\end{theorem}

\begin{proof}
Since \(\rho\) and \(\sigma\) are Hermitian,
\[
[\rho,\sigma]^\dagger=\sigma\rho-\rho\sigma.
\]
Expanding the squared Hilbert--Schmidt norm and using cyclicity of the trace,
\begin{align*}
\Tr([\rho,\sigma]^\dagger[\rho,\sigma])
&=\Tr\bigl((\sigma\rho-\rho\sigma)(\rho\sigma-\sigma\rho)\bigr)\\
&=\Tr(\sigma\rho^2\sigma)-\Tr(\sigma\rho\sigma\rho)\\
&\quad -\Tr(\rho\sigma\rho\sigma)+\Tr(\rho\sigma^2\rho)\\
&=2\Tr(\rho^2\sigma^2)-2\Tr(\rho\sigma\rho\sigma).
\end{align*}
Therefore \(\Gap(\rho,\sigma)=\frac12\|[\rho,\sigma]\|_2^2\). Since the
Hilbert--Schmidt norm is faithful, \(\Gap(\rho,\sigma)=0\) if and only if
\([\rho,\sigma]=0\). For a finite Hermitian family, simultaneous
diagonalizability is equivalent to pairwise commutativity. Applying the
two-state criterion to every unordered pair proves the finite-family
statement.
\end{proof}

The theorem gives a decision rule. It also gives a quantitative invariant. For
a finite family, define the total Bargmann gap
\[
\BG(\vec\rho)=\sum_{i<j}\Gap(\rho_i,\rho_j).
\]
Then
\[
\BG(\vec\rho)=0
\quad\Longleftrightarrow\quad
\vec\rho\ \text{is set incoherent}.
\]
Thus \(\BG\) provides a faithful basis-independent indicator of the pairwise
noncommutativity that prevents a common incoherent basis.

The spectral form of the gap explains what the fourth-order data are detecting.
Let
\[
\rho=\sum_i\lambda_i|i\rangle\langle i|
\]
be the spectral decomposition of \(\rho\), and write
\(\sigma_{ij}=\langle i|\sigma|j\rangle\) in this eigenbasis. Then
\[
(\rho\sigma)_{ij}=\lambda_i\sigma_{ij},
\qquad
(\sigma\rho)_{ij}=\lambda_j\sigma_{ij},
\]
so
\(
[\rho,\sigma]_{ij}=(\lambda_i-\lambda_j)\sigma_{ij}.
\)
Hence
\[
\|[\rho,\sigma]\|_2^2
=\sum_{i,j}(\lambda_i-\lambda_j)^2|\sigma_{ij}|^2.
\]
The diagonal terms vanish. Since \(\sigma\) is Hermitian, the double sum counts
each unordered pair twice, and therefore
\[
\Gap(\rho,\sigma)
=\sum_{i<j}(\lambda_i-\lambda_j)^2|\sigma_{ij}|^2.
\]
Thus the gap detects exactly the coherences of \(\sigma\) between eigenspaces
of \(\rho\) with different eigenvalues. Matrix elements inside a degenerate
eigenspace do not contribute, as they do not obstruct simultaneous
diagonalization with \(\rho\).

The same formula also makes the noise behavior transparent. For depolarized
states
\[
\rho_p=(1-p)\rho+p\frac{\Id_d}{d},
\qquad
\sigma_q=(1-q)\sigma+q\frac{\Id_d}{d},
\]
the identity part commutes with every operator, and hence
\[
[\rho_p,\sigma_q]=(1-p)(1-q)[\rho,\sigma].
\]
Consequently,
\[
\Gap(\rho_p,\sigma_q)=(1-p)^2(1-q)^2\Gap(\rho,\sigma).
\]
Depolarizing noise therefore suppresses the fourth-order signal by a known
multiplicative factor.

\section{Discussion and Outlook}
\label{sec:conclusion}

We have determined a low-order Bargmann-invariant hierarchy for set coherence.
We begin  with second-order data. For qubits, the Hilbert--Schmidt
lengths and overlap are sufficient because commutativity is equivalent to
collinearity of Bloch vectors. In dimension three, this metric information is
no longer enough. Third-order data then become sufficient for qutrit pairs,
because they determine the spectra and enforce the required block structure.
The same mechanism fails in dimension four, where even the full pairwise
order-\(\le3\) data can miss noncommutativity.

The first dimension-independent pairwise criterion appears at fourth order.
The two cyclic words \(1122\) and \(1212\) compare different orderings of the
same two states, and their difference is exactly one half of the squared
Hilbert--Schmidt norm of the commutator. This gives a compact and
basis-independent criterion for two-state set coherence in every finite
dimension. Applied to all unordered pairs, it yields a complete
Bargmann-invariant test for arbitrary finite families.

Conceptually, the hierarchy shows that different Bargmann orders capture
different layers of relational information. Second-order traces see metric
overlap. Third-order traces can capture enough spectral information in low
dimension. Fourth-order ordering-sensitive traces are the first pairwise
invariants that universally detect the noncommutativity responsible for set
coherence.

Several directions remain open. One may seek smaller complete scenarios under
additional assumptions, such as fixed spectra, fixed ranks, known degeneracy
patterns, or restricted state families. It would
also be interesting to develop analogous low-order Bargmann hierarchies for
other relational quantum resources, including set imaginarity, measurement
incompatibility, and channel incompatibility~\cite{Uola2020,Heinosaari2016}.

\vskip 5pt
\noindent\emph{Note add:} Shortly after submitting our paper to Physical Review A, we recently became aware of a related article on commutativity involving a single equality with Bargmann invariants of order 4  (see Ref. \cite{Wagner2026bBargmannScenarios}).
\vskip 5pt

\begin{acknowledgments}
This work was supported by the Guangdong Basic and Applied Basic Research
Foundation under Grant No. 2024A1515030023 and the National Natural Science
Foundation of China under Grant No. 12371458.
\end{acknowledgments}

\appendix

\section{Many-State Lower-Order Limitations}
\label{app:many}

The lower-order negative results in the main text were stated for two states.
Here we record analogous limitations for finite families. These examples show
that the lower-order failures are not artifacts of the two-state setting.

\subsection{Pairwise second-order data for qutrit families}

Let
\[
\W_2^{(n)}=\{ij:1\le i\le j\le n\},
\]
which records all pairwise Hilbert--Schmidt inner products, including
purities.

\begin{proposition}
For every \(n\ge2\),
\[
C_3(\W_2^{(n)})\cap I_3(\W_2^{(n)})\neq\varnothing.
\]
Thus all pairwise second-order data can be identical for a jointly diagonal
qutrit family and for a set-coherent qutrit family.
\end{proposition}

\begin{proof}
Let
\[
D_1=\frac1{\sqrt2}
\begin{pmatrix}
1&0&0\\
0&-1&0\\
0&0&0
\end{pmatrix},
\qquad
D_2=\frac1{\sqrt6}
\begin{pmatrix}
1&0&0\\
0&1&0\\
0&0&-2
\end{pmatrix}.
\]
These matrices are traceless, Hermitian, Hilbert--Schmidt orthonormal, and
commuting. Define
\[
A=D_1,
\qquad
B=\frac1{\sqrt2}
\begin{pmatrix}
0&0&1\\
0&0&0\\
1&0&0
\end{pmatrix}.
\]
Then \(A,B\) are traceless, Hermitian, and Hilbert--Schmidt orthonormal, but
\[
[A,B]\neq0.
\]

Choose real vectors
\[
r_k=(a_k,b_k)\in\mathbb R^2,
\qquad k=1,\ldots,n,
\]
not all lying on one line. For sufficiently small \(\epsilon>0\), define
\[
\rho_k^{\mathrm{cl}}
=\frac{\Id_3}{3}+\epsilon(a_kD_1+b_kD_2),
\]
and
\[
\rho_k^{\mathrm q}
=\frac{\Id_3}{3}+\epsilon(a_kA+b_kB).
\]
All these matrices are density operators when \(\epsilon\) is small enough.
The first family is commuting. The second is not commuting, because
\[
[\rho_i^{\mathrm q},\rho_j^{\mathrm q}]
=\epsilon^2(a_ib_j-b_ia_j)[A,B],
\]
which is nonzero for some \(i,j\).

For all \(i,j\),
\[
\Tr(\rho_i^{\mathrm{cl}}\rho_j^{\mathrm{cl}})
=\frac13+\epsilon^2(a_ia_j+b_ib_j),
\]
and the same formula holds for \(\rho_i^{\mathrm q},\rho_j^{\mathrm q}\).
Hence the two families have identical \(\W_2^{(n)}\)-data, while one is set
incoherent and the other is set coherent.
\end{proof}

\subsection{Pairwise third-order data for four-dimensional families}

Let
\[
\W_{\le3}^{(n)}
=\bigcup_{1\le i<j\le n}
\{ii,jj,ij,iii,jjj,iij,ijj\}.
\]

\begin{proposition}
For every \(n\ge2\),
\[
C_4(\W_{\le3}^{(n)})\cap I_4(\W_{\le3}^{(n)})\neq\varnothing.
\]
Thus all pairwise words of order at most three can be identical for a jointly
diagonal four-dimensional family and for a set-coherent four-dimensional
family.
\end{proposition}

\begin{proof}
Let
\begin{align*}
D_1&=\frac1{\sqrt2}\diag(1,-1,0,0),\\
D_2&=\frac1{\sqrt2}\diag(0,0,1,-1).
\end{align*}
Let
\[
A=\frac12
\begin{pmatrix}
1&0&0&0\\
0&1&0&0\\
0&0&-1&0\\
0&0&0&-1
\end{pmatrix},
\qquad
B=\frac12
\begin{pmatrix}
0&0&1&0\\
0&0&0&1\\
1&0&0&0\\
0&1&0&0
\end{pmatrix}.
\]
Equivalently,
\[
A=\frac{\sigma_z\otimes\Id_2}{2},
\qquad
B=\frac{\sigma_x\otimes\Id_2}{2}.
\]
Both pairs have the same quadratic trace tensor and vanishing cubic trace
tensor:
\[
\Tr(D_\alpha D_\beta)=\Tr(A_\alpha A_\beta),
\]
and
\[
\Tr(D_\alpha D_\beta D_\gamma)
=\Tr(A_\alpha A_\beta A_\gamma)=0
\]
for all \(\alpha,\beta,\gamma\in\{1,2\}\), where \((A_1,A_2)=(A,B)\).

Choose real vectors \(r_k=(a_k,b_k)\in\mathbb R^2\) not all collinear and set
\[
X_k^{\mathrm{cl}}=a_kD_1+b_kD_2,
\qquad
X_k^{\mathrm q}=a_kA+b_kB.
\]
For sufficiently small \(\epsilon>0\), define
\[
\rho_k^{\mathrm{cl}}=\frac{\Id_4}{4}+\epsilon X_k^{\mathrm{cl}},
\qquad
\rho_k^{\mathrm q}=\frac{\Id_4}{4}+\epsilon X_k^{\mathrm q}.
\]
The first family is diagonal. The second is set coherent because
\[
[\rho_i^{\mathrm q},\rho_j^{\mathrm q}]
=\epsilon^2(a_ib_j-b_ia_j)[A,B],
\]
which is nonzero for some \(i,j\). Expanding
\[
\rho_k=\frac{\Id_4}{4}+\epsilon X_k
\]
shows that every word in \(\W_{\le3}^{(n)}\) depends only on the corresponding
quadratic and cubic trace tensors. These tensors agree for the classical and
quantum constructions above. Therefore the two families have identical
\(\W_{\le3}^{(n)}\)-data.
\end{proof}

These examples show that the lower-order failures persist even when all
pairwise low-order words are collected. A universal pairwise Bargmann
characterization of set coherence cannot stop below fourth order.

\bibliography{bargmann_set_coherence_refs}

\end{document}